\documentclass{ifacconf}

\usepackage{graphicx}      %
\usepackage{natbib}        %

\usepackage[retainorgcmds]{IEEEtrantools}

\usepackage{physics2}
\usephysicsmodule{ab, diagmat, xmat}

\newcommand{\bmat}[1]{\begin{bmatrix}#1\end{bmatrix}}

\newcommand{\norm}[1]{\ab\|#1\|} 

\usepackage[separate-uncertainty=true]{siunitx}
\usepackage{amssymb}
\usepackage{arydshln}
\usepackage{amsmath}
\usepackage{accents}
\usepackage{mathtools}

\usepackage[
    bb=stixtwo, 
    scr=euler,
]{mathalpha}

\usepackage{url}

\usepackage{enumitem}

\usepackage{cleveref}
\crefname{assumption}{assumption}{assumptions}
\Crefname{assumption}{Assumption}{Assumptions}
\crefname{condition}{condition}{conditions}
\Crefname{condition}{Condition}{Conditions}
\crefname{figure}{Fig.}{Figs.}
\Crefname{figure}{Fig.}{Figs.}
\crefname{equation}{}{}
\Crefname{equation}{}{}

\crefname{lem}{Lemma}{Lemmas}
\Crefname{lem}{Lemma}{Lemmas}
\crefname{rem}{Remark}{Remarks}
\Crefname{rem}{Remark}{Remarks}
\crefname{thm}{Theorem}{Theorems}
\Crefname{thm}{Theorem}{Theorems}
\crefname{cor}{Corollary}{Corollaries}
\Crefname{cor}{Corollary}{Corollaries}
\crefname{prop}{Proposition}{Propositions}
\Crefname{prop}{Proposition}{Propositions}
\crefname{defn}{Definition}{Definitions}
\Crefname{defn}{Definition}{Definitions}

\crefname{app}{Appendix}{Appendices}
\Crefname{app}{Appendix}{Appendices}

\makeatletter
\patchcmd{\@IEEEyesnumber}
  {\stepcounter}
  {\refstepcounter}
  {}{}
\patchcmd{\@@IEEEeqnarray}
  {\stepcounter}
  {\refstepcounter}
  {}{}
\patchcmd{\@@IEEEeqnarraycr}
  {\stepcounter{IEEEsubequation}}
  {\refstepcounter{IEEEsubequation}}
  {}{}
\patchcmd{\@@IEEEeqnarraycr}
  {\stepcounter{IEEEsubequation}}
  {\refstepcounter{IEEEsubequation}}
  {}{}
\patchcmd{\@@IEEEeqnarraycr}
  {\stepcounter{IEEEequation}}
  {\refstepcounter{IEEEequation}}
  {}{}
\patchcmd{\@@IEEEeqnarraycr}
  {\stepcounter{IEEEequation}}
  {\refstepcounter{IEEEequation}}
  {}{}
\makeatother

\usepackage{nicematrix}
\NiceMatrixOptions{margin}

\usepackage{booktabs} %
\usepackage{multirow}
\usepackage{makecell}

\usepackage{array}
\newcommand{\PreserveBackslash}[1]{\let\temp=\\#1\let\\=\temp}
\newcolumntype{C}[1]{>{\PreserveBackslash\centering}m{#1}}
\newcolumntype{R}[1]{>{\PreserveBackslash\raggedleft}m{#1}}
\newcolumntype{L}[1]{>{\PreserveBackslash\raggedright}m{#1}}
\usepackage{nicematrix}
\usepackage{arydshln}

\usepackage{tikz}
\usepackage{circuitikz}
\usetikzlibrary{calc,shapes,arrows.meta,quotes,positioning,angles,quotes, decorations.pathmorphing,patterns,decorations.markings,
backgrounds,fit}

\usepackage{todonotes}
\setuptodonotes{inline, inlinewidth=\linewidth}

\renewcommand{\L}{\mathsf{L}}
\newcommand{\R}{\mathsf{R}}
\newcommand{\tf}[1]{\mathbf{#1}}
\newcommand{\ttf}[1]{\boldsymbol{#1}}
\DeclareMathOperator{\spreg}{SpReg_{2}}
\DeclareMathOperator*{\H2}{\mathcal{H}_2}
\DeclareMathOperator*{\Hinfty}{\mathcal{H}_\infty}

\DeclareMathOperator*{\graph}{\mathcal{G}}
\DeclareMathOperator*{\vertices}{\mathcal{V}}
\DeclareMathOperator*{\edges}{\mathcal{E}}

\definecolor{NavyBlue}{RGB}{20,20,200}
\newcommand{\red}[1]{\textcolor{red}{#1}}

\definecolor{blue_connection_color}{RGB}{115, 141, 211}
\newcommand{\blue}[1]{\textcolor{blue_connection_color}{#1}}

\definecolor{yellow_nodes}{RGB}{206, 157, 53}
\newcommand{\yellow}[1]{\textcolor{yellow_nodes}{#1}}

\begin{document}
\begin{frontmatter}

\title{%
Data-Driven Optimal Distributed Controller Synthesis via Spatial Regret
}

\thanks[footnoteinfo]{This work was supported as a part of NCCR Automation, a National Centre of Competence in Research, funded by the Swiss National Science Foundation (grant number 51NF40\_225155), the Swiss National Science Foundation (grant number 200021-204962), the Swiss National Science Foundation Ambizione (grant  number PZ00P2\_208951) and the NECON project (grant number 200021-219431).}

\author[EPFL]{Vaibhav Gupta\thanksref{equal_contributions}} 
\author[EPFL]{Daniele Martinelli\thanksref{equal_contributions}} 
\author[EPFL]{Giancarlo Ferrari-Trecate}
\author[Oxford]{Luca Furieri}
\author[EPFL]{Alireza Karimi}

\address[EPFL]{Laboratoire d’Automatique, EPFL, CH-1015 Lausanne, Switzerland (e-mail: $\{$vaibhav.gupta, daniele.martinelli, giancarlo.ferraritrecate, alireza.karimi$\}$@epfl.ch)}
\address[Oxford]{Department of Engineering Science, University of Oxford, United Kingdom (email: luca.furieri@eng.ox.ac.uk)}

\thanks[equal_contributions]{V. Gupta and D. Martinelli contributed equally to this work.}

\begin{abstract}
In this paper, we present a novel method for synthesising an optimal distributed spatial regret controller using experimentally obtained frequency-response data. \emph{Spatial regret} provides a measure of the performance gap between a structured distributed controller and an oracle with enhanced communication topology. We relax assumptions on the communication topology, allowing the oracle to adopt any enhanced structure. %
While this generalisation requires an iterative solution in place of a single convex program, we provide a tractable algorithm that synthesises optimal controllers from frequency-response data while preserving stability and the desired communication structure.
Through numerical examples, we illustrate the better performance of the spatial regret controller compared to classical $\mathcal{H}_2/\mathcal{H}_\infty$ designs, underscoring the effectiveness of the proposed methodology.
\end{abstract}


\end{frontmatter}

\section{Introduction}
As systems grow in scale and complexity, centralised control becomes ineffective due to single points of failure, high computational burden, and limited communication bandwidth \citep{scattolini2009architectures}. Distributed architectures address these challenges by assigning local controllers to subsystems that coordinate through limited communication, with applications ranging from power systems to cooperative robotics \citep{dorfler2014sparsity,ren2008distributed}.
However, designing optimal distributed controllers under structural communication constraints remains notoriously challenging \citep{witsenhausen1968counterexample}.
Building on this observation, several works have identified conditions under which distributed controllers for linear systems can be synthesised via convex optimisation, most notably through quadratic invariance and its extensions~\citep{rotkowitz2005characterization,furieri2020sparsity,naghnaeian2024youla}.

A fundamental question in optimal control concerns how to evaluate controller performance in the presence of disturbances. Controllers are typically evaluated via closed-loop norms (${\H2/\Hinfty}$) \citep{zhou1998essentials}, which embed assumptions on disturbance characteristics: the $\H2$ norm assumes disturbances generated by stochastic processes, whereas the $\Hinfty$ norm considers worst-case disturbances. However, these metrics are graph-agnostic and treat all disturbances uniformly, failing to capture how performance degradation depends on the location and structure of the disturbances. In distributed control, where structural constraints fundamentally restrict achievable disturbance rejection, it is natural to ask how much performance is lost due to limited information and communication.

\citet{martinelli2025spatialregret} addressed this question by introducing \emph{spatial regret}, a metric that quantifies the performance gap between a constrained distributed controller and an \emph{oracle} operating under an augmented communication graph. The oracle serves as a hypothetical (``what-if'') reference benchmark controller with additional sensing or communication links, which we aim to mimic.
For example, in a networked system where certain nodes experience disturbances, a suitable oracle topology could enable more distant nodes to access measurements from these affected nodes, thereby reducing disturbance propagation. The spatial regret quantifies the performance penalty induced by communication constraints in the actual network topology, revealing which links, if added, would most effectively improve performance. However, the synthesis framework in \citet{martinelli2025spatialregret} assumes an accurate state-space model, which is often unavailable or unreliable for large-scale systems.

This motivates a \emph{data-driven} approach, where controllers are designed directly from frequency-response measurements without requiring an intermediate system identification step. The use of frequency response for the analysis and synthesis of controllers in linear systems has been a well-established practice in the literature. The frequency response of the system can be easily extracted from the input-output data, as presented in \citet{pintelon2012system}. However, controller design using frequency-domain data typically results in a non-convex optimisation problem. A non-smooth optimisation framework was used in \citet{apkarian_StructuredH_Infty_2018} to compute fixed-structure  $\mathcal{H}_\infty$ controllers directly from frequency-domain data. Several methodologies employ convex approximations, such as PID design using constraint linearization \citep{hast_PIDDesignConvexconcave_2013}, fixed-structure controller design for SISO systems \citep{karimi_FixedorderH_Infty_2010}, and MIMO-PID tuning via convex-concave optimisation \citep{boyd_MIMOPIDTuning_2016}. \citet{karimi_DatadrivenApproachRobust_2017} proposed a fixed-structure data-driven controller design method for multivariable systems that achieves mixed $\mathcal{H}_2 / \mathcal{H}_\infty$ sensitivity performance using iterative convex optimisations. Furthermore, \citet{schuchert_DatadrivenFixedstructureFrequencybased_2024} extended the technique to fixed-structure controller design for generalised systems in linear fractional representation (LFR) form. \citet{gupta_NonparametricIQCMultipliers_2026} integrated robustness into the framework for the non-linearity expressed as Integral Quadratic Constraints (IQCs). Additionally, the framework facilitates controller design for distributed control architectures and has been successfully implemented in microgrids \citep{madani_DataDrivenDistributedCombined_2021}.

This paper develops a data-driven framework for spatial regret minimisation, enabling controller synthesis directly from frequency-response measurements. We extend \citet{martinelli2025spatialregret} by recasting spatial regret minimisation within the data-driven frequency-domain framework of \citet{schuchert_DatadrivenFixedstructureFrequencybased_2024}, removing the requirement for a parametric model and generalising the admissible controller structures. While this flexibility comes at the cost of requiring an iterative solution instead of a single convex program and of losing completeness over the set of all stabilising controllers, the resulting framework enables practical implementation without requiring a parametric model.

Our contribution is threefold. First, we extend the well-posedness result to arbitrary controller classes, removing restrictive assumptions on controller parameterisation. Second, we reformulate spatial regret minimisation in the frequency domain and utilise a convex lower bound that yields a tractable iterative algorithm for fixed-structure distributed controllers. Third, we establish that this algorithm preserves closed-loop stability at each iteration when initialised with a stabilising controller. 
Numerical experiments on a power-grid model demonstrate that data-driven spatial regret controllers achieve superior performance for localised disturbances compared to classical norm-based designs.

\textbf{\textit{Notation:}} 
In this article, ${M \succ (\succeq)  \, N}$ indicates that ${M-N}$ is a positive (semi-) definite matrix and ${M \prec(\preceq) \, N}$ indicates ${M-N}$ is negative (semi-) definite. The zero and identity matrices of appropriate size are denoted $0$ and $I$, respectively. The transpose of a matrix $M$ is denoted by $M^T$ and its conjugate transpose by $M^*$. When $M \in \mathbb{C}^{n\times m}$ is full row rank, its right inverse is defined as $M^\R:= M^* (M M^*)^{-1}$, satisfying $M M^\R = I$ and $M^\R M$ are Hermitian. Conversely, when $M$ is full column rank, then its left inverse is defined as $M^\L := (M^* M)^{-1}M^*$, satisfying $M^\L M = I$ and  $M M^\L$ is Hermitian. In the case $M$ is square and full rank, $M^\R = M^\L = M^{-1}$. 
$\mathcal{R}_p^{n_o \times n_i}$ denotes the set of real-rational proper transfer function matrices for a MIMO system with $n_o$ outputs and $n_i$ inputs, while $\mathcal{RH}_\infty^{n_o \times n_i}$ denotes the set of real rational stable transfer functions with bounded infinity norm. For discrete-time systems, let $z$ denote the Z-transform variable and ${\Omega \coloneqq [0, \pi/T_s)}$, where $T_s$ is the sampling time.
We denote transfer function in bold (e.g., $\tf{M}$).
For any matrix \(M\) that is partitioned into blocks, we denote by \(M^{[i,j]}\) the \((i,j)\) block of \(M\), with subblock dimensions determined by context. We use similar notation for vectors and transfer function matrices: $x^{[i]}$ is the $i$-th subvector of $x$, and $\tf{M}^{[i,j]}$ denotes the $(i,j)$ block of $\tf{M}$.
We use $\norm{\cdot}_2$ to denote the $\ell_2$ norm of signals and, by slight abuse of notations, we also use $\norm{\cdot}_2$ and $\norm{\cdot}_\infty$ to refer to the $\H2$ and $\Hinfty$ system norms, respectively.
Additionally, we consider a directed graph represented as $\graph = (\vertices, \edges)$, where the set of vertices is denoted as $\vertices$ and the set of edges as $\edges$.

\section{Problem formulation} \label{sec:problem_formulation}
We consider networked discrete-time linear time-invariant (LTI) systems whose interaction topology is described by a directed graph \(\graph\). Each node $i$ of the graph represents a subsystem with dynamics given by:
\begin{equation}\label{eq:ss_networked_system}
\begin{aligned}
x^{[i]}_{t+1} &= \sum_{j \in \mathcal{N}_i} A^{[i,j]} x^{[j]}_t + B^{[i,j]}_{1} w^{[j]}_t  + B^{[i,j]}_{2} u^{[i]}_t, 
\\
z^{[i]}_t &= \sum_{j \in \mathcal{N}_i} C^{[i,j]}_{1} x^{[j]}_t + D^{[i,j]}_{11} w^{[j]}_t + D^{[i,j]}_{12} u^{[i]}_t,
\\
y^{[i]}_t &= \sum_{j \in \mathcal{N}_i} C^{[i,j]}_{2} x^{[j]}_t + D^{[i,j]}_{21} w^{[j]}_t + D^{[i,j]}_{22} u^{[i]}_t,
\end{aligned}
\end{equation}
where \(x^{[i]}_t \in \mathbb{R}^{n_i}\) is the state, \(u^{[i]}_t \in \mathbb{R}^{m_i}\) the control input, \(w^{[i]}_t \in \mathbb{R}^{{n_w}_i}\) the disturbance, \(y^{[i]}_t \in \mathbb{R}^{p_i}\) the measured output, and \(z^{[i]}_t \in \mathbb{R}^{{n_z}_i}\) the performance output.
The overall system dimensions are \(n = \sum_{i=1}^N n_i\), \(m = \sum_{i=1}^N m_i\), \(p = \sum_{i=1}^N p_i\), \({n_w} = \sum_{i=1}^N {n_w}_i\), and \({n_z} = \sum_{i=1}^N {n_z}_i\).
In this work, we assume the state-space matrices in~\eqref{eq:ss_networked_system} are unknown, and we develop a data-driven control design that operates directly on frequency response data without requiring parametric model identification. For a plant $G$ as in~\eqref{eq:ss_networked_system}, the input-output behaviour is captured by its Frequency Response Function (FRF) $\tf{G}$:
\begin{equation*}
    \begin{bmatrix}
        \tf{z} \\ \tf{y}
    \end{bmatrix}
    =
    \underbrace{\begin{bmatrix}
        \tf{G}_{11} & \tf{G}_{12}
        \\
        \tf{G}_{21} & \tf{G}_{22}
    \end{bmatrix}}_{\tf{G}}
    \begin{bmatrix}
        \tf{w} \\ \tf{u}
    \end{bmatrix}\,,
\end{equation*}
where $\tf{G}_{ij} = C_i (z I - A )^{-1} B_j + D_{ij}$.
We can estimate the plant transfer function \(\tf{G}_{22}(e^{j\omega})\) from input-output measurements using standard Fourier analysis. From $m$ experiments with persistently exciting inputs, we obtain \citep{pintelon2012system}:
\begin{equation}\label{eq:plant_from_data}
\tf{G}_{22}(e^{j\omega}) =
\left[
\sum_{k=0}^{N_s-1} Y_k e^{-j\omega T_s k}
\right]
\left[
\sum_{k=0}^{N_s-1} U_k e^{-j\omega T_s k}
\right]^{-1},
\end{equation}
where \(N_s\) is the number of samples per experiment, \(T_s\) is the sampling period, each column of \(U_k \in \mathbb{R}^{m \times m}\) and \(Y_k \in \mathbb{R}^{p \times m}\) contains the input and output at time \(k\) from one experiment, and \(\tf{G}_{22}(e^{j\omega})\) is evaluated at a discrete set of frequencies \(\omega \in \Omega_{N_s} \subset [0,\pi/T_s)\). Moreover, the synthesis process can account for estimation errors due to truncation and noise in the plant’s frequency response. \(\tf{G}_{11}(e^{j\omega})\), \(\tf{G}_{12}(e^{j\omega})\), and \(\tf{G}_{21}(e^{j\omega})\) encode the desired performance specifications and the disturbances to be rejected, and are defined by the user in conjunction with the estimated \(\tf{G}_{22}(e^{j\omega})\). Next, we review data-driven control design approaches that leverage this frequency response information. These methods account for truncation and noise-induced errors inherent in this estimation procedure.

\subsection{Data-Driven Controller Synthesis} \label{subsec:data_driven_classical_results}

We impose the following conditions on the plant $\tf{G}$:
\begin{enumerate}[label=(A\arabic*), leftmargin=3em]
    \item \label[assumption]{assumption:1b} $\tf{G}_{12}(e^{j\omega})$ has full column rank $\forall\omega\in\Omega$.
    \item \label[assumption]{assumption:2} $\tf{G}(e^{j\omega})$ is bounded for all $\omega \in \Omega$.
\end{enumerate}
These assumptions are standard in data-driven controller synthesis (e.g., \citep{schuchert_DatadrivenFixedstructureFrequencybased_2024}).
\Cref{assumption:1b} ensures that all control inputs influence the performance channels.
\Cref{assumption:2} excludes plants with poles on the stability boundary. This technical condition simplifies the theoretical development but can be relaxed using small detours along the Nyquist contour \citep[Remark 1]{schuchert_DatadrivenFixedstructureFrequencybased_2024}.

We seek to design controllers that ensure closed-loop stability while optimising performance under fixed communication constraints. Motivated by optimality results for closed-loop norm minimisation \citep{zhou1998essentials}, we consider linear dynamical output-feedback controllers of the form $\tf{u} = \tf{K} \tf{y}$, where the controller structure is restricted to a prescribed class:
\begin{equation*}
    \tf{K} \in \mathcal{K} \subseteq \mathcal{R}_p^{ n_o \times n_i},
\end{equation*}
and $ \mathcal{K}$ denotes the set of admissible fixed-structure distributed controllers. This structure encodes communication constraints that dictate which measurements are available to which control actuators. 
\begin{exmp}
    \label[example]{example:distributed_controller}
    Consider controllers for 3 subsystems with a communication graph as in \Cref{fig:example1}. The admissible set $\mathcal{K}$ for the ${3 \times 3}$ controller takes the form:
    \begin{equation} \label{eq:example_K_sparse_prob_form}
        \mathcal{K} = \ab\{
            \tf{K} \in \mathcal{R}_p^{3\times 3} 
            \ \middle\mid \ 
            \tf{K} =
            \begin{bmatrix}
            \bullet & \bullet & \tf{0} \\
            \bullet & \bullet & \bullet\\
            \tf{0}  & \bullet & \bullet
            \end{bmatrix}
        \},
    \end{equation}
    where $\bullet$ denotes an arbitrary transfer function and $\tf{0}$ indicates prohibited information flow. Note that this framework naturally accommodates communication delays by incorporating delay operators in the structure of $\tf{K}$. 
\end{exmp}
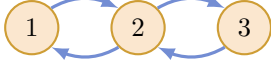
\begin{figure}
    \centering
    \begin{tikzpicture}[
        auto,  
        node distance=4em,
        >=latex,
        thick,
    ]
    \definecolor{nodefill_color}{RGB}{251, 231, 207}
    \definecolor{nodedraw_color}{RGB}{206, 157, 53}
    \definecolor{connection_color}{RGB}{115, 141, 211}
    \tikzset{
        node/.style={
            circle, 
            fill=nodefill_color, draw=nodedraw_color, text=black, 
            minimum height=2em, minimum width=2em, 
            },
    }
    
    \node[node] (N1) {$1$};
    \node[node, right of=N1] (N2) {$2$};
    \node[node, right of=N2] (N3) {$3$};

    \foreach \k in {2,...,3}{
        \pgfmathsetmacro{\prev}{int(\k - 1)};
        \draw[->, connection_color, very thick, out=30, in=150]
            (N\prev.north east) to (N\k.north west);
        \draw[<-, connection_color, very thick, out=-30, in=-150]
            (N\prev.south east) to (N\k.south west);
    }

\end{tikzpicture}
    \caption{Communication graph for \Cref{example:distributed_controller}.}
    \label{fig:example1}
\end{figure}

For zero initial conditions ${x_0 = 0}$ and bounded disturbances $w$ with ${\norm{w}_2 < \infty}$, we measure controller performance through the induced $\ell_2$ norm of the performance output: ${\norm{{z}}_2^2}$. Define the closed-loop transfer function matrix from $w$ to $z$ as the lower fractional transformation (LFT) of $\tf{G}$ \citep{zhou1998essentials}:
\begin{equation}\label{eq:definition_LFT}
\tf{T}_{zw} \coloneq \tf{G}_{11} + \tf{G}_{12} \tf{K}(I - \tf{G}_{22}\tf{K})^{-1} \tf{G}_{21}\,.
\end{equation}
The cost ${\norm{{z}}_2}$ can be expressed as a function of $\tf{K}$ and the disturbance $w$. Specifically, we define the cost function $J({w},\tf{K}) \coloneq \norm{z}_2^2$.
Since the disturbance $w$ is unmeasurable and unknown during controller design, minimising $J(w,\tf{K})$ for all $w$ is infeasible \citep{goel2023regret}. Classical control theory addresses this challenge by minimising ${\mathbb{E}_{ w \sim \mathcal{N}(0 , I )} [J (w, \tf{K}) ] }$ and ${\max_{\norm{w}_2 \leq 1} [J (w, \tf{K})]}$, which correspond to the $\H2$ and $\Hinfty$ norms of $\tf{T}_{zw}$, respectively. 
The classical optimal distributed control problem is therefore:
\begin{equation} \label{eq:classical_optimal_distributed_control_problem}
    \min_{ \tf{K} \in \mathcal{C}_{\text{stab}} \cap \mathcal{K} } \,  \norm{\tf{T}_{zw}}_{2/\infty}\,,
\end{equation}
where $\mathcal{C}_{\text{stab}} \coloneqq \ab\{\tf{K} \mid \tf{K} \text{ stabilises } \tf{G}_{22}\}$ denotes the set of all stabilising controllers.
We synthesise $\mathcal{H}_2$ and $\mathcal{H}_\infty$ controllers using the frequency-domain approach of \citet{schuchert_DatadrivenFixedstructureFrequencybased_2024}. For a stabilising controller, the system norms admit the frequency-domain representations:
\begin{align*}
    \norm{\tf{T}_{zw}}_2^2 
        &= \frac{1}{2\pi} \int_\Omega \operatorname{trace}\ab( \tf{T}_{zw}^*(e^{j\omega}) \tf{T}_{zw}(e^{j\omega}) ) d\omega \,,
    \\
    \norm{\tf{T}_{zw}}_\infty^2 
        &= \sup\limits_{\omega\in\Omega}\, \overline{\sigma} \ab( \tf{T}_{zw}^*(e^{j\omega}) \tf{T}_{zw}(e^{j\omega}) ) \,,
\end{align*}
where $\overline{\sigma}(\cdot)$ denotes the maximum singular value. Then, problem~\eqref{eq:classical_optimal_distributed_control_problem} can be formulated as the minimisation of an upper bound on the system norms 
\begin{IEEEeqnarray*}{l'rCl'r} 
    \IEEEeqnarraymulticol{5}{c}{
        \min\limits_{\Gamma, \, \tf{K} \in \mathcal{C}_{\text{stab}} \cap \mathcal{K} }\, \gamma
    } \IEEEyesnumber %
\\ \text{s.t.,} 
    & \tf{T}^*_{zw}(e^{j\omega}) \tf{T}_{zw}(e^{j\omega}) &\preceq& \Gamma(e^{j\omega}), & \forall\omega\in\Omega,
\end{IEEEeqnarray*}
where $\Gamma(e^{j\omega})$ is a Hermitian matrix. For $\mathcal{H}_\infty$ minimisation, we set $ \Gamma(e^{j\omega})=\gamma I$ with $\gamma \in \mathbb{R}$. For $\mathcal{H}_2$ minimisation, we have:
\begin{IEEEeqnarray}{rCl}
    \gamma &=& \frac{1}{2\pi}\int_\Omega \text{trace}\left( \Gamma(e^{j\omega}) \right) d\omega.
\end{IEEEeqnarray}

Following \citet{schuchert_DatadrivenFixedstructureFrequencybased_2024}, we express the controller as a left factorisation $\tf{K} = \tf{Y}^{-1} \tf{X}$ and solve a sequence of convex optimisations, where each iteration uses the stabilising controller ${\tf{K}_c = \tf{Y}_c^{-1} \tf{X}_c }$ from the previous step:
\begin{IEEEeqnarray*}{rCl}
    \IEEEeqnarraymulticol{3}{c}{
        \min\limits_{\Gamma, \, \tf{X} \in \mathcal{X}, \, \tf{Y} \in \mathcal{Y} }\,
            {\gamma}
    }
    \IEEEyesnumber \label{eq:datadriven_H2_Hinf}
\\
    \bmat{
        \Gamma - \ab(\tf{\Psi} \tf{G}_{11})^* \ab(\tf{\Psi} \tf{G}_{11})
            & \ab( \tf{\Phi} \tf{G}_{11} + X \tf{G}_{21} )^* \\
        \ab( \tf{\Phi} \tf{G}_{11} + X \tf{G}_{21} )
            & \tf{\Phi}_c \tf{\Phi}^* + \tf{\Phi} \tf{\Phi}_c^* - \tf{\Phi}_c \tf{\Phi}_c^*
    } &\succeq& 0, \,\forall\omega
\end{IEEEeqnarray*}
where:
\begin{align*}
    \tf{\Phi} &\coloneq \ab( \tf{Y} - \tf{X} \tf{G}_{22} ) \tf{G}_{12}^\L,
\\
    \tf{\Phi}_c &\coloneq \ab( \tf{Y}_c - \tf{X}_c \tf{G}_{22} ) \tf{G}_{12}^\L,
\\        
    \tf{\Psi} &\coloneq I - \tf{\Phi}^\R \tf{\Phi} = I - \tf{G}_{12} \tf{G}_{12}^\L.
\end{align*}
Here, the requirement that ${\tf{K} \in \mathcal{K}}$ is ensured by the appropriate selection of the classes ${\mathcal{X}\subseteq\mathcal{RH}_\infty}$ and ${\mathcal{Y}\subseteq\mathcal{RH}_\infty}$. We give an example of one such class from \citet{schuchert_DatadrivenFixedstructureFrequencybased_2024}. Consider the same set of ${3 \times 3}$ controllers in \Cref{example:distributed_controller}, then the classes $\mathcal{X}$ and $\mathcal{Y}$ can be chosen as follows:
\begin{subequations}
\begin{align}
    \mathcal{Y} &= \ab\{ \tf{Y} \mid \tf{Y} = C_Y(z I - A_Y)^{-1} B_Y + D_Y \} \\
    \mathcal{X} &= \ab\{ \tf{X} \mid \tf{X} = C_X(z I - A_X)^{-1} B_X + D_X \}
\end{align}
\end{subequations}
with:
\begin{align*}
    A_Y = A_X &= \bmat{
        \diamond & 0 & 0 \\
        0 & \diamond & 0 \\
        0 & 0 & \diamond \\
    } &
    C_Y = C_X &= \bmat{
        \diamond & 0 & 0 \\
        0 & \diamond & 0 \\
        0 & 0 & \diamond \\
    } \\
    B_Y &= \bmat{
        \bullet & 0 & 0 \\
        0 & \bullet & 0 \\
        0 & 0 & \bullet \\
    } &
    D_Y &= \bmat{
        \bullet & 0 & 0 \\
        0 & \bullet & 0 \\
        0 & 0 & \bullet \\
    } \\
    B_X &= \bmat{
        \bullet & \bullet & 0       \\
        \bullet & \bullet & \bullet \\
        0       & \bullet & \bullet \\ 
    } &
    D_X &= \bmat{
        \bullet & \bullet & 0       \\
        \bullet & \bullet & \bullet \\
        0       & \bullet & \bullet \\ 
    }
\end{align*}
Here, since we want $\mathcal{X}$ and $\mathcal{Y}$ to be linearly parametrised, $\bullet$~are chosen as the optimisation variables of appropriate dimensions, and $\diamond$~as some fixed matrices of appropriate dimensions.

It can be shown that the closed-loop stability is preserved throughout the iterations: the synthesised controller ${\tf{K} \in \mathcal{C}_{\text{stab}}}$ if the initial controller ${\tf{K}_c \in \mathcal{C}_{\text{stab}}}$ \citep{schuchert_DatadrivenFixedstructureFrequencybased_2024}. For a stable plant $\tf{G}_{22}$, an initial stabilising controller can be obtained by a sufficiently small gain. For unstable plants, a stabilising controller is required for system identification and serves as the initial iterate. However, it is important to emphasise that this closed-loop stability condition is only a sufficient condition, i.e., it represents a subset of $\mathcal{C}_{\text{stab}}$.

\subsection{Spatial regret}
While classical $\H2$ and $\Hinfty$ metrics provide valuable design objectives, they are inherently graph-agnostic: they treat all disturbances equally, irrespective of the network constraints imposed on the controller.

We adopt a different design philosophy introduced by \citet{martinelli2024closing, martinelli2025spatialregret}: prioritise robustness against disturbances for which information limitations constitute the primary bottleneck. To formalise this principle, we compare performance against an oracle controller $\tf{\hat{K}}$ with enhanced communication topology, where $\tf{\hat{K}} \in \mathcal{C}_{\text{stab}} \cap \hat{\mathcal{K}}$ and $\hat{\mathcal{K}} \supseteq \mathcal{K}$. The superset $\hat{\mathcal{K}}$ represents a relaxation of the original structural constraint set $\mathcal{K}$. If $\tf{\hat{K}}$ is designed to minimise a closed-loop norm over $\hat{\mathcal{K}}$, this richer information structure may allow the oracle to outperform any distributed controller constrained to $\mathcal{K}$. The \emph{spatial regret} metric quantifies the performance gap as:
\begin{equation}\label{eq:def_spatial_regret_cost}
        \spreg(\tf{K},\tf{\hat{K}}) 
        \coloneq
        \max_{\norm{{w}}_2 \leq 1} 
        \left[
        J({w},\tf{K}) - J({w},\tf{\hat{K}}) \right]\,,
\end{equation}
where 
$J({w},\tf{\hat{K}})$ denotes the oracle's performance, with:
\begin{equation}\label{eq:def_oracle_LFT}
\tf{\hat{T}}_{zw} \coloneq \tf{G}_{11} + \tf{G}_{12} \tf{\hat{K}}(I - \tf{G}_{22}\tf{\hat{K}})^{-1} \tf{G}_{21}\,.
\end{equation}
The performance gap $\left[J(w,\tf{K}) - J(w,\tf{\hat{K}})\right]$ quantifies the performance loss due to information constraints for a disturbance $w$, and the spatial regret captures the worst-case degradation over all bounded disturbances. Given an oracle $\tf{\hat{K}}$, we synthesise the distributed controller $\tf{K}$ that minimises spatial regret while preserving closed-loop stability:
\begin{equation} \label{eq:min_problem_spregret_generic_time}
    \min_{\tf{K}\in \mathcal{C}_{\textnormal{stab}} \cap \mathcal{K}} \spreg(\tf{K},\tf{\hat{K}})\,.
\end{equation}
Although spatial regret is naturally defined in the time domain via~\eqref{eq:def_spatial_regret_cost}, the following result establishes an equivalent frequency-domain characterisation that is directly amenable to data-driven synthesis.
\begin{lem}[\citet{martinelli2025spatialregret}]\label[lem]{lem:spreg_frequency_domain}
    Let $\tf{K}, \tf{\hat{K}} \in \mathcal{R}_{p}^{n_o \times n_i}$, then:
    \begin{equation} \label{eq:sup_lambda_max_Psi}
    \spreg(\tf{K}, \tf{\hat{K}}) = \sup_{\omega \in \Omega} \lambda_{\text{max}}
    ( \tf{\Lambda} (e^{j\omega}) ),
    \end{equation}
    where $\lambda_{\text{max}}(\tf{\Lambda}(e^{j\omega}))$ denotes the largest eigenvalue of:
\begin{equation}\label{eq:def_Psi_ejomega}
    \tf{\Lambda} (e^{j\omega}) \coloneq \tf{T}_{zw}^*(e^{j\omega}) \tf{T}_{zw}(e^{j\omega})-\tf{\hat{T}}_{zw}^*(e^{j\omega}) \tf{\hat{T}}_{zw}(e^{j\omega})\,.
\end{equation}
\end{lem}

This result enables reformulation of problem~\eqref{eq:min_problem_spregret_generic_time} in the frequency domain:
\begin{equation} \label{eq:min_problem_spregret_generic_freq}
    \min_{\tf{K}\in \mathcal{C}_{\textnormal{stab}} \cap \mathcal{K}} 
    \:
    \sup_{\omega \in \Omega} 
    \:
    \lambda_{\text{max}}
    ( \tf{\Lambda} (e^{j\omega}) )
\end{equation}

Not every oracle, however, yields a meaningful benchmark: we require that the oracle cannot be uniformly outperformed by controllers with more restrictive communication constraints, which would invalidate it as a benchmark. This leads to the following well-posedness requirement.

\begin{defn}\label{def:wellposedness}
Given $\hat{\mathcal{K}} \supseteq \mathcal{K}$, the spatial regret metric is \emph{well-posed} for a given oracle $\tf{\hat{K}} \in  \mathcal{C}_{\text{stab}} \cap \mathcal{\hat{K}} $ if:
\begin{equation} \label{eq:def_wellposedness}
    \spreg(\tf{K},\tf{\hat{K}}) \geq 0\,, \quad \forall\, \tf{K}\in \mathcal{C}_{\text{stab}} \cap \mathcal{K} \,.
\end{equation}
\end{defn}

When condition~\eqref{eq:def_wellposedness} is violated, some constrained controller outperforms the oracle, rendering the benchmark invalid and the regret metric meaningless. For the model-based setting and for some specific controller classes $\mathcal{K}$ and $\hat{\mathcal{K}}$, \citet{martinelli2025spatialregret} provided sufficient conditions for well-posedness and reformulated spatial regret minimisation as a semidefinite program.

In this work, we do not restrict controller classes to those considered in~\cite{martinelli2025spatialregret}, precluding the direct application of their model-based approach. Instead, we develop a data-driven synthesis framework that operates on frequency-response data alone. In the next section, we establish two main contributions. First, we extend the well-posedness result of~\cite{martinelli2025spatialregret} to arbitrary controller classes $\mathcal{K}$ and $\hat{\mathcal{K}}$ (Lemma~\ref{lemma:well_posedness}), removing restrictive assumptions on controller structure. Second, we derive a tractable data-driven synthesis procedure that computes fixed-structure distributed controllers minimising~\eqref{eq:min_problem_spregret_generic_freq} from frequency-response measurements, while ensuring closed-loop stability (\Cref{thm:convexification}).

The generalisation to arbitrary controller structures comes at a trade-off: in contrast to \citet{martinelli2025spatialregret}, our methodology doesn't ensure complete exploration of the stabilising controller set $\mathcal{C}_{\text{stab}}$. In terms of controller structures, our approach utilises the left factorisation, whereas \citet{martinelli2025spatialregret} utilises FIR-based controllers. This makes the direct comparison of their set sizes challenging. Additionally, the synthesis procedure yields a semi-infinite program, which we solve through frequency gridding. While the grid can be refined to arbitrarily high resolution, this increases computational cost accordingly.

\begin{rem}
The spatial regret framework extends naturally to incorporate frequency-domain weighting filters $\tf{W} \in \mathcal{RH}_\infty$ that encode prior knowledge about disturbance spectral characteristics. Replacing $\tf{T}_{zw}$ with $(\tf{T}_{zw}\tf{W})$ in~\eqref{eq:definition_LFT} allows emphasising performance in specific frequency ranges, as usually done in $\mathcal{H}_2$ and $\mathcal{H}_\infty$ design~\citep{zhou1998essentials}.
\end{rem}

\section{Main results} \label{sec:main_results}
We first establish conditions under which the spatial regret metric is well-posed (cf.~\Cref{def:wellposedness}). 
While \citet[Th.~1]{martinelli2025spatialregret} provided well-posedness guarantees for the Youla-Ku\v{c}era parameterisation with specific structured controller sets, we extend this result to the general direct controller parameterisation setting with arbitrary constraint sets $\mathcal{K}$ and $\hat{\mathcal{K}}$.

\begin{lem}\label[lem]{lemma:well_posedness}
Let $\hat{\mathcal{K}} \supseteq \mathcal{K}$ and let the oracle $\tf{\hat{K}}$ be computed as the solution of:
\begin{equation}\label{eq:def_oracle_problem_in_K}
    \min_{ \tf{K} \in \mathcal{C}_{\text{stab}} \cap \mathcal{\hat{K}} }\,
        \norm{\tf{T}_{zw}}_{2/\infty}.
\end{equation}
Then $\spreg (\tf{K}, \tf{\hat{K}}) \geq 0$ for any $\tf K \in \mathcal{C}_{\text{stab}} \cap \mathcal{K}$.
\end{lem}
\begin{pf}
Given $\tf{\hat{K}}$ the solution of~\eqref{eq:def_oracle_problem_in_K}, we prove that:
\begin{equation}\label{eq:th_1_step_1}
        \forall \tf{\bar{K}} \in \mathcal{C}_{\text{stab}}\cap \hat{\mathcal{K}}, \:\:\exists\norm{{w}}_2\leq 1 \: : \:J( {w} ,\tf{\bar{K}}) \geq 
        J({w},\tf{\hat{K}})\,.
    \end{equation}
We proceed by contradiction. Suppose there exists $\tf{\bar{K}} \in \mathcal{C}_{\text{stab}}\cap \hat{\mathcal{K}}$ such that:
    \begin{equation}\label{eq:contradicted_sentence_well_posed}
        J({w},\tf{\bar{K}}) < J({w},\tf{\hat{K}}), \quad \forall \norm{{w}}_2\leq 1.
    \end{equation} 
If the oracle minimises the $\H2$ norm, it achieves optimality against the orthonormal set of impulse disturbances~\cite[Ch.~4.3]{zhou1998essentials}:
    \begin{equation*}
       \norm{\tf{\hat{T}}_{zw} }_{2}^2 = \sum_{i=1}^{m}\norm{\hat{z}_i}_2^2,
       \end{equation*}
where $\hat{z}_i$ is the performance output when an impulse is applied to the $i^{th}$ disturbance channel of~\eqref{eq:def_oracle_LFT}. Since $\tf{\hat{K}}$ is $\H2$-optimal, it follows that:
    \begin{equation}\label{eq:step_proof_h2_oracle_is_well_posed}
        \sum_{i=1}^{m} \norm{\hat{z}_i}_2^2 \leq \sum_{i=1}^{m} \norm{\bar{z}_i}_2^2,
    \end{equation}
where $\bar{z}_i$ is the performance output when the same impulse is applied using $\tf{\bar{K}}$. From~\eqref{eq:step_proof_h2_oracle_is_well_posed}, there must exist an index $i^* \in \{1, \ldots, m\}$ such that $\|{\hat{z}}_{i^*}\|_2^2 \leq \|{\bar{z}}_{i^*}\|_2^2$. Therefore, for the impulse disturbance applied to channel $i^*$, we have $J({w}_{i^*}, \tf{\hat{K}}) \leq J({w}_{i^*}, \tf{\bar{K}})$, which contradicts~\eqref{eq:contradicted_sentence_well_posed}.

Let $w^* \in \arg\max_{\lVert {w} \rVert_{2} \leq 1}J({w}, \tf{\bar{K}})$. If the oracle is $\Hinfty$ optimal, then for any $\norm{w}_2 \leq 1$:
\begin{equation}\label{eq:intermediate_step_Hinf_proof_well_posed}
        J(w^*, \tf{\bar{K}})
        = 
        \max_{\lVert {w} \rVert_{2} \leq 1}J({w}, \tf{\bar{K}})
        \geq 
        \max_{\lVert {w} \rVert_{2} \leq 1}J({w}, \tf{\hat{K}})
        \geq
        J(w^*, \tf{\hat{K}}),
    \end{equation}
which contradicts~\eqref{eq:contradicted_sentence_well_posed}.
Hence,~\eqref{eq:th_1_step_1} holds for any oracle obtained as in~\eqref{eq:def_oracle_problem_in_K}. Since $\mathcal{C}_{\text{stab}}\cap \mathcal{K} \subseteq \mathcal{C}_{\text{stab}}\cap \hat{\mathcal{K}}$, it follows that:
    \begin{equation}\label{eq:last_step_firts_preposition}
        \forall \tf{K} \in \mathcal{C}_{\text{stab}}\cap \mathcal{K}, \:\:\exists\norm{{w}}_2\leq 1 \: : \:J( {w} ,\tf{K}) \geq 
        J({w},\tf{\hat{K}})\,.
     \end{equation}
Using the definition~\eqref{eq:def_spatial_regret_cost} of spatial regret,~\eqref{eq:last_step_firts_preposition} establishes that $\spreg(\tf{K}, \tf{\hat{K}}) \geq 0$ for any $\tf{K} \in \mathcal{C}_{\text{stab}} \cap \mathcal{K}$.
\qed
\end{pf}

\Cref{lemma:well_posedness} establishes that oracles designed to minimise classical $\H2$ or $\Hinfty$ norms over the enhanced constraint set $\hat{\mathcal{K}}$ guarantee well-posedness of the spatial regret metric, providing a principled approach for oracle selection in practice. 
Practically, in a data-driven manner, the oracle $\tf{\hat{K}}$ can be computed by leveraging the methods from \Cref{subsec:data_driven_classical_results}, in particular by solving problem~\eqref{eq:datadriven_H2_Hinf} where we enforce that $\tf{\hat{K}} \in \mathcal{C}_{\text{stab}} \cap \hat{\mathcal{K}}$.

With a well-posed oracle $\tf{\hat{K}}$ in hand, we now develop the data-driven synthesis framework for controllers that minimise the spatial regret metric.

\begin{lem} \label{th:initial_reformulation_spreg}
    Consider an oracle ${\tf{\hat{K}} \in \mathcal{C}_{\text{stab}} \cap \hat{\mathcal{K}}}$ minimising~\eqref{eq:def_oracle_problem_in_K}.
    The left-factorised controller ${ \tf{K} = \tf{Y}^{-1} \tf{X} }$ which is a feasible solution to the problem~\eqref{eq:min_problem_spregret_generic_freq} can be obtained by solving:
    \begin{IEEEeqnarray*}{rCl'l}
    \IEEEeqnarraymulticol{4}{c}{
        \min \limits_{\gamma, \, \tf{X} \in \mathcal{X}, \, \tf{Y} \in \mathcal{Y}} \, {\gamma}
    }  \IEEEyesnumber \label{eq:base_optimisation}
    \\ 
        \bmat{
        \tf{\hat{\Gamma}} - \ab(\tf{\Psi} \tf{G}_{11})^* \ab(\tf{\Psi} \tf{G}_{11})
            & \ab( \tf{\Phi} \tf{G}_{11} + X \tf{G}_{21} )^* \\
        \ab( \tf{\Phi} \tf{G}_{11} + X \tf{G}_{21} )
            & \tf{\Phi} \tf{\Phi}^*
        } (e^{j\omega}) &\succeq& 0, & \forall\omega
    \\
    \IEEEeqnarraymulticol{4}{c}{
        \tf{Y}^{-1} \tf{X} \in \mathcal{C}_{\text{stab}}
    }
    \end{IEEEeqnarray*}
    where ${\mathcal{X}\subseteq\mathcal{RH}_\infty}$ and ${\mathcal{Y}\subseteq\mathcal{RH}_\infty}$ are chosen such that ${\tf{K}  \in \mathcal{K}}$, and:
    \begin{align*}
        \tf{\Phi} & \coloneq \ab( \tf{Y} - \tf{X} \tf{G}_{22} ) \tf{G}_{12}^\L ,
    \\        
        \tf{\Psi} & \coloneq I - \tf{\Phi}^\R \tf{\Phi} = I - \tf{G}_{12} \tf{G}_{12}^\L,
    \\
        \tf{\hat{\Gamma}} & \coloneq \gamma I + \tf{\hat{T}}_{zw}^* \tf{\hat{T}}_{zw},
    \end{align*}
    with $\tf{\hat{T}}_{zw}$  defined in~\eqref{eq:def_oracle_LFT}.
\end{lem}
\begin{pf}
    Define an auxiliary scalar $\gamma$ to bound the largest singular value of $\tf{\Lambda}(e^{j\omega})$ in~\eqref{eq:min_problem_spregret_generic_freq}. This results in an equivalent epigraph formulation:
    \begin{IEEEeqnarray*}{rCl'l}
    \IEEEeqnarraymulticol{4}{c}{
        \min \limits_{\tf{K} \in \mathcal{C}_\text{stab} \cap \mathcal{K}, \, \gamma} \, {\gamma}
    }  \IEEEyesnumber \label{eq:spreg_base_optimisation}
    \\ 
        \tf{T}_{zw} (e^{j\omega})^* \tf{T}_{zw} (e^{j\omega})
            - \tf{\hat{T}}_{zw} (e^{j\omega})^* \tf{\hat{T}}_{zw} (e^{j\omega}) 
            &\preceq& \gamma I, & \forall\omega.
    \end{IEEEeqnarray*}
    Any controller ${\tf{K} \in \mathcal{K}}$ can be expressed through left factorisation as ${\tf{K} = \tf{Y}^{-1} \tf{X}}$, with ${X\in\mathcal{X}\subset\mathcal{RH}_\infty}$ and ${Y\in\mathcal{Y}\subset\mathcal{RH}_\infty}$ as defined in \Cref{subsec:data_driven_classical_results} (see \Cref{app:equivalence} for more details). So the closed-loop transfer function becomes:
    \begin{equation}
        \tf{T}_{zw} = \tf{\Phi}^\R \ab( \tf{\Phi} \tf{G}_{11} + \tf{X} \tf{G}_{21}) + \tf{\Psi} \tf{G}_{11} 
    \end{equation}
    
    Since ${\tf{\Psi}^* \tf{\Phi}^\R  = \tf{\Psi} \tf{\Phi}^\R = \tf{\Phi}^\R - \tf{\Phi}^\R \tf{\Phi} \tf{\Phi}^\R = 0 }$, the constraint becomes:
    \begin{multline}
    \label{eq:spreg_constraint_preschur}
    \tf{T}_{zw}^* \tf{T}_{zw}
        = \ab( \tf{\Phi} \tf{G}_{11} + \tf{X} \tf{G}_{21} )^*
          \ab( \tf{\Phi} \tf{\Phi}^* )^\R \ab( \tf{\Phi} \tf{G}_{11} + \tf{X} \tf{G}_{21} )
    \\
        + \ab( \tf{\Psi} \tf{G}_{11} )^* \ab( \tf{\Psi} \tf{G}_{11}  )
        \preceq \gamma I + \tf{\hat{T}}_{zw}^* \tf{\hat{T}}_{zw}.
    \end{multline}
    Because $\tf{G}_{12}$ has full column rank by assumption, $ \tf{\Phi} \tf{\Phi}^* $ is a full-rank square matrix and ${\ab( \tf{\Phi} \tf{\Phi}^* )^\R = \ab( \tf{\Phi} \tf{\Phi}^* )^{-1}}$. Applying the Schur complement lemma to \eqref{eq:spreg_constraint_preschur} yields:
    \begin{equation}
        \bmat{
        \tf{\hat{\Gamma}} - \ab(\tf{\Psi} \tf{G}_{11})^* \ab(\tf{\Psi} \tf{G}_{11})
            & \ab( \tf{\Phi} \tf{G}_{11} + X \tf{G}_{21} )^* \\
        \ab( \tf{\Phi} \tf{G}_{11} + X \tf{G}_{21} )
            & \tf{\Phi} \tf{\Phi}^*
        } (e^{j\omega}) \succeq 0.
    \end{equation}
    Finally, we impose that $\tf{Y}^{-1} \tf{X} \in \mathcal{C}_{\text{stab}}$ to ensure the closed-loop stability, i.e. ${\tf{K}\in\mathcal{C}_\text{stab}}$. So, putting all this together, we obtain problem~\eqref{eq:base_optimisation}.
    \qed
\end{pf}

\Cref{th:initial_reformulation_spreg} requires \Cref{assumption:1b} to guarantee the existence of $\tf{G}_{12}^\L$. 
Problem~\eqref{eq:base_optimisation} is nonconvex due to the presence of the quadratic term $\tf{\Phi}\tf{\Phi}^*$, where $\tf{\Phi}$ depends affinely on $(\tf{X},\tf{Y})$.
Furthermore, it is not obvious how to guarantee the closed-loop stability using the obtained controller, i.e. ${\tf{Y}^{-1} \tf{X} \in \mathcal{C}_\text{stab} }$, making problem~\eqref{eq:base_optimisation} non-tractable.
In that effect, we propose the next theorem, which uses an iterative convexification of problem~\eqref{eq:base_optimisation} while guaranteeing closed-loop stability at each iteration.

\begin{thm}\label{thm:convexification}
    A locally optimal solution for problem~\eqref{eq:base_optimisation} can be obtained through iterative optimisation around an initial stabilising controller ${\tf{K}_c = \tf{Y}_c^{-1} \tf{X}_c } \in \mathcal{C}_{\textnormal{stab}} \cap \mathcal{K}$ by solving:
    \begin{IEEEeqnarray*}{l'rCl'l}
    \IEEEeqnarraymulticol{5}{c}{
        \min \limits_{\tf{X}, \, \tf{Y}, \, \gamma} \, {\gamma}
    }  \IEEEyesnumber \label{eq:spreg_convex_optimisation}
    \\ 
        & 
        \bmat{
        \tf{\hat{\Gamma}} - \ab(\tf{\Psi} \tf{G}_{11})^* \ab(\tf{\Psi} \tf{G}_{11})
            & \ab( \tf{\Phi} \tf{G}_{11} + X \tf{G}_{21} )^* \\
        \ab( \tf{\Phi} \tf{G}_{11} + X \tf{G}_{21} )
            & \tf{\Phi}_c \tf{\Phi}^* + \tf{\Phi} \tf{\Phi}_c^* - \tf{\Phi}_c \tf{\Phi}_c^*
        }&\succeq& 0, & \forall\omega
    \end{IEEEeqnarray*}
    where ${\tf{\Phi}_c = \ab( \tf{Y}_c - \tf{X}_c \tf{G}_{22} ) \tf{G}_{12}^\L }$. 
    Moreover, any solution of \eqref{eq:spreg_convex_optimisation} is a stabilising controller.
\end{thm}
\begin{pf}
    We can derive a convex lower bound for $\tf{\Phi} \tf{\Phi}^*$ by expanding the positive semidefinite condition given by ${(\tf{\Phi} - \tf{\Phi}_c) (\tf{\Phi} - \tf{\Phi}_c)^* \succeq 0}$ for any choice of $\tf{\Phi}_c$:
    \[
        \tf{\Phi} \tf{\Phi}^* \succeq \tf{\Phi}_c \tf{\Phi}^* + \tf{\Phi} \tf{\Phi}_c^* - \tf{\Phi}_c \tf{\Phi}_c^*.
    \]
    The choice of $\tf{\Phi}_c$ is critical for guaranteeing closed-loop stability. 
    By \citet[Theorem 1]{schuchert_DatadrivenFixedstructureFrequencybased_2024}, selecting ${\tf{\Phi}_c = \ab( \tf{Y}_c - \tf{X}_c \tf{G}_{22} ) \tf{G}_{12}^\L }$ with a stabilising controller ${\tf{K}_c = \tf{Y}_c^{-1} \tf{X}_c }$ ensures that any solution ${\tf{K} = \tf{Y}^{-1} \tf{X} }$ of problem~\eqref{eq:spreg_convex_optimisation} also stabilises the closed-loop system.
    \qed
\end{pf}

\Cref{thm:convexification} provides a tractable solution to the spatial regret minimisation problem~\eqref{eq:min_problem_spregret_generic_freq}.
Unlike the method in \citet{martinelli2025spatialregret}, our approach imposes no restrictions on the controller sets $\mathcal{K}$ and $\hat{\mathcal{K}}$ beyond the structural constraints themselves, substantially increasing design flexibility. The cost of this generalisation is that problem~\eqref{eq:spreg_convex_optimisation} must be solved iteratively, with convergence rate depending on the initial controller and problem conditioning~\citep{schuchert2023achieving}.
Furthermore, the closed-loop stability provided in \Cref{thm:convexification} is a sufficient condition; a necessary and sufficient condition remains to be determined.

\begin{rem}
    It is interesting to note that the primary difference in the constraint set of the controller synthesis problem among $\H2$, $\Hinfty$, and spatial regret lies in the choice of $\tf{\Gamma}$. This observation aligns with the intuition that spatial regret represents a compromise, influenced by the oracle, between the $\H2$ and $\Hinfty$ at each frequency point.
\end{rem}

\begin{rem}
    Similarly to the classical data-driven methods discussed in~\Cref{subsec:data_driven_classical_results}, the requirement of a pre-stabilising controller $\tf{K}_c$ is natural in data-driven settings, as safe and informative data collection for identifying $\tf{G}$ as in~\eqref{eq:plant_from_data} typically necessitates stabilisation.
\end{rem}

\section{Numerical Results} \label{sec:numerical_results}

We illustrate the data-driven spatial regret synthesis on the swing dynamics of a 5-bus power grid, following the model in~\citet{anderson2019system, martinelli2025spatialregret}. We first describe the plant and the frequency-domain data used for controller synthesis, then compare the spatial regret controller against classical $\H2$ and $\Hinfty$ distributed designs.

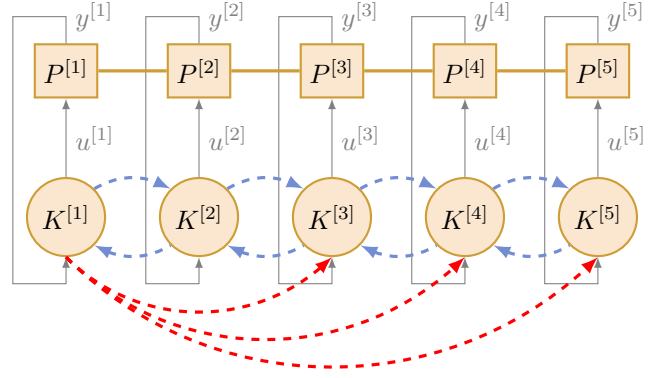
\begin{figure}[t]
    \centering
    \begin{tikzpicture}[
        auto,  
        node distance=5em,
        >=latex,
        thick,
    ]
    \definecolor{nodefill_color}{RGB}{251, 231, 207}
    \definecolor{nodedraw_color}{RGB}{206, 157, 53}
    \definecolor{connection_color}{RGB}{115, 141, 211}
    \tikzset{
        plant/.style={
            rectangle, 
            fill=nodefill_color, draw=nodedraw_color, text=black, 
            minimum height=2em, minimum width=2em, 
            },
        node/.style={
            circle, 
            fill=nodefill_color, draw=nodedraw_color, text=black, 
            minimum height=2em, minimum width=2em, 
            },
    }
    
    \node[node] (K1) {$K^{[1]}$};
    \node[node, right of=K1] (K2) {$K^{[2]}$};
    \node[node, right of=K2] (K3) {$K^{[3]}$};
    \node[node, right of=K3] (K4) {$K^{[4]}$};
    \node[node, right of=K4] (K5) {$K^{[5]}$};

    \node[plant, above of=K1, yshift=0.5em] (P1) {$P^{[1]}$};
    \node[plant, right of=P1] (P2) {$P^{[2]}$};
    \node[plant, right of=P2] (P3) {$P^{[3]}$};
    \node[plant, right of=P3] (P4) {$P^{[4]}$};
    \node[plant, right of=P4] (P5) {$P^{[5]}$};

    \foreach \k in {1,...,5}{
        \draw[->, gray, thin] (K\k.north) -- (P\k.south) 
            node[midway, right] {$u^{[\k]}$};
        \draw[->, gray, thin] (P\k.north) 
            |- +(-2em, 1em) node[midway, right] {$y^{[\k]}$}
            -- ($(K\k.south) + (-2em, -1em)$)
            -| (K\k.south); 
    }

    \foreach \k in {2,...,5}{
        \pgfmathsetmacro{\prev}{int(\k - 1)};
        \draw[-, nodedraw_color, very thick]
            (P\prev.east) -- (P\k.west);
    }

    \foreach \k in {2,...,5}{
        \pgfmathsetmacro{\prev}{int(\k - 1)};
        \draw[->, connection_color, dashed, very thick, out=30, in=150]
            (K\prev.north east) to (K\k.north west);
        \draw[<-, connection_color, dashed, very thick, out=-30, in=-150]
            (K\prev.south east) to (K\k.south west);
    }

    \draw[->, dashed, red, very thick, out=-45, in=-135]
        (K1.south) to (K3.south);
    \draw[->, dashed, red, very thick, out=-45, in=-135]
        (K1.south) to (K4.south);
    \draw[->, dashed, red, very thick, out=-45, in=-135]
        (K1.south) to (K5.south);
\end{tikzpicture}
    \caption{Interaction graph of the 5-bus power system model. Blue dashed arrows show the communication connections for the controller. Red dashed arrows indicate the oracle's additional connections.
    }
    \label{fig:grid_scheme}
\end{figure}

We consider the discrete-time swing dynamics of the 5-bus network in Fig.~\ref{fig:grid_scheme}. 
Each bus $i \in \{1,\ldots,5\}$ is a subsystem with state $x^{[i]}_t \in \mathbb{R}^2$ representing phase angle and frequency deviations. For simplicity, all buses share identical parameters (in per-unit): inertia $m_i = \SI{2}{}$, damping $d_i = \SI{2}{}$, and coupling $k_{ij} = \SI{20}{}$. The sampling period is $T_s = \SI{0.02}{s}$. The network is open-loop stable, thus data can be collected with $\tf{K}_c = 0$. Each subsystem evolves according to~\eqref{eq:ss_networked_system} with matrices:
\begin{gather*}
    A^{[i,i]} = \begin{bmatrix}
        1 & T_s \\
        -\frac{k_i}{m_i}T_s & 1 - \frac{d_i}{m_i}T_s
    \end{bmatrix}, \:\:
            A^{[i,j]} = \begin{bmatrix}
        0 & 0 \\
        \frac{k_{ij}}{m_i}T_s & 0
    \end{bmatrix},
    \\
    B_{1}^{[i,i]} = \begin{bmatrix}
        0\\
        1
    \end{bmatrix}
    ,
    \:
        B_{2}^{[i,i]} = \begin{bmatrix}
            0 \\ \frac{T_s}{m_i}
        \end{bmatrix},
    \:
        C_{1}^{[i,i]} = \begin{bmatrix}
            1 \\ 0
        \end{bmatrix}^T,
    \:
        C_{2}^{[i,i]} = \begin{bmatrix}
            1 \\ 0
        \end{bmatrix}^T,
    \\
        D_{12}^{[i,i]} = \begin{bmatrix}
            0 \\ 1
        \end{bmatrix},
    \quad
        D_{21}^{[i,i]} = 1,  
    \quad
    D_{11} = D_{22} = 0\,.
\end{gather*}

For the structure of the distributed controllers we impose delayed nearest-neighbour communication (one-step delay) and immediate local feedback:
\begin{equation*}    
\mathcal{K} = \ab\{
        \tf{K} \in \mathcal{R}_p^{5\times 5} 
            \ \middle\mid \ 
        \tf{K} =
        \begin{bmatrix}
            \yellow{\bullet} & \frac{1}{z}\blue{\bullet} & \tf{0}  & \tf{0} & \tf{0} \\
            \frac{1}{z}\blue{\bullet} & \yellow{\bullet} & \frac{1}{z}\blue{\bullet} & \tf{0} & \tf{0} \\
            \tf{0}  & \frac{1}{z}\blue{\bullet} & \yellow{\bullet} & \frac{1}{z}\blue{\bullet} & \tf{0} \\
            \tf{0}  & \tf{0}  & \frac{1}{z}\blue{\bullet} & \yellow{\bullet} & \frac{1}{z}\blue{\bullet} \\
            \tf{0}  & \tf{0}  & \tf{0}  & \frac{1}{z}\blue{\bullet} & \yellow{\bullet}
        \end{bmatrix}
    \}\,,
\end{equation*}
where $\bullet$ denotes a scalar transfer function of order 2. Hence, $u^{[i]}_t$ depends on $y^{[i]}_t$ and delayed neighbour signals $y^{[j]}_{t-1}$, where $j$ are the neighbours of $i$.

Our framework aims to develop a controller that, by mimicking a graph-informed oracle, achieves better performance compared to traditional control techniques.
To investigate this, we assume that the first bus is prone to disturbances. Hence, we synthesise an oracle $\tf{\hat{K}}$ in which each subcontroller has immediate access to the measurements of the $1^{st}$ bus. In \Cref{fig:grid_scheme}, the additional connections in the oracle graph are shown in red. We also assume each subcontroller to access immediately measurements from neighbours. Hence we define:
\begin{equation*}
    \hat{\mathcal{K}} = \ab\{
        \tf{\hat{K}} \in \mathcal{R}_p^{5\times 5} 
        \ \middle\mid \ 
        \tf{\hat{K}} =
        \begin{bmatrix}
        \yellow{\bullet} & \blue{\bullet} & \tf{0}  & \tf{0} & \tf{0} \\
        \blue{\bullet} & \yellow{\bullet} & \blue{\bullet} & \tf{0} & \tf{0} \\
        \red{\bullet}  & \blue{\bullet} & \yellow{\bullet} & \blue{\bullet} & \tf{0} \\
        \red{\bullet}  & \tf{0}  & \blue{\bullet} & \yellow{\bullet} & \blue{\bullet} \\
        \red{\bullet}  & \tf{0}  & \tf{0}  & \blue{\bullet} & \yellow{\bullet}
        \end{bmatrix}
    \}\,.
\end{equation*}
We compute the oracle $\tf{\hat{K}} \in \hat{\mathcal{K}}$ by solving~\eqref{eq:datadriven_H2_Hinf} for $\Hinfty$ minimisation. The spatial regret controller $\tf{K}^{\textnormal{SR}} \in \mathcal{K}$ is obtained as in~\eqref{eq:spreg_convex_optimisation}. 
For benchmarks, we synthesise $\tf{K}^{2}$ and $\tf{K}^{\infty}$ by solving~\eqref{eq:datadriven_H2_Hinf} with $\H2$ and $\Hinfty$ objectives, respectively and with $\tf{K}^{2}, \tf{K}^{\infty} \in \mathcal{K}$.
All designs use $600$ logarithmically spaced frequency samples in $[10^{-2}, \frac{\pi}{T_s}]$.

To assess performance, we evaluate $\|\tf{T}_{zw}(e^{j\omega})\|_2^2$ for each $\omega$, which equals the worst-case cost over multi-sinusoidal disturbances at frequency $\omega$ \citep[Ch.~4]{zhou1998essentials}:
\begin{equation*}
    \norm{\tf{T}_{zw}(e^{j \omega})}_2^2 = \max_{\{w_i\}, \{\phi_i\}} \frac{\norm{z}_2^2}{\norm{w}_2^2}
\end{equation*}
where ${w = \begin{bmatrix}
    w_1 \cos(\omega t + \phi_1)   & \cdots & w_5 \cos(\omega t + \phi_5)   
\end{bmatrix}^T
}$ such that $\sum_{i=1}^{5} w_i = 1$. 
As shown in \Cref{fig:plot_with_5_masses_on_all_channels}, $\tf{K}^{\textnormal{SR}}$ achieves a favourable trade-off between the average performance of $\tf{K}^2$ and the robust performance of $\tf{K}^\infty$ over the resonance peak around \SI{8}{\radian\per\sec}. 

However, these results consider disturbances applied to all buses simultaneously. Due to the oracle design, we expect $\tf{K}^{\text{SR}}$ to perform better for disturbances localised at bus 1. To investigate this, we apply single-bus disturbances of the form:
\begin{equation*}
    \bar{w}_t \coloneq \begin{bmatrix}
        \cos(\omega t) & 0 & 0 & 0 & 0
    \end{bmatrix}^T\,,
\end{equation*}
and we denote with $\bar{z}_t$ its corresponding measured output. Hence, for each controller and for each frequency $\omega \in \Omega$, we estimate $\norm{\tf{T}_{zw}^{[:,1]}(e^{j \omega})}_2^2$, which satisfies:
\begin{equation*}
    \norm{\tf{T}_{zw}^{[:,1]}(e^{j \omega})}_2^2 =     \begin{cases}
        \langle \norm{\bar{z}}^2 \rangle & \textnormal{if } \omega = 0 \textnormal{ or } \omega = \pm \frac{\pi}{T_s}
        \\
        2\:\langle \norm{\bar{z}}^2 \rangle & \textnormal{otherwise}\,,
    \end{cases}
\end{equation*}
where $\langle \norm{\bar{z}}^2 \rangle \coloneq \lim_{T\to \infty} \frac{1}{T} \sum_{t=0}^{T-1} \norm{\bar{z}_t}^2$ is the time-averaged energy of the regulated output when $w_t = \bar{w}_t$.
\Cref{fig:plot_with_5_masses_on_1_channel} shows $\norm{\tf{T}_{zw}^{[:,1]}(e^{j \omega})}_2^2$ as a function of frequency.
The spatial regret controller $\tf{K}^{\textnormal{SR}}$ closely mimics the oracle $\tf{\hat{K}}$, achieving improved performance compared to both $\tf{K}^2$ (at the peak around \SI{8}{\radian\per\sec}) and $\tf{K}^\infty$ (across all other frequencies).
For time-domain validation, we simulate a multi-sine disturbance with frequencies \SIlist{8;38}{\radian\per\sec} applied to bus 1.
\Cref{fig:response_5_masses_sinusoidal_disturbance} shows the time evolution of $\norm{z_t}^2$.
The spatial regret controller achieves an average $\norm{z}_2$ reduction of \SI{21.72}{\percent} relative to $\tf{K}^2$ and \SI{48.00}{\percent} relative to $\tf{K}^\infty$.

\begin{figure}[t]
    \centering
    \includegraphics[width=0.9\linewidth]{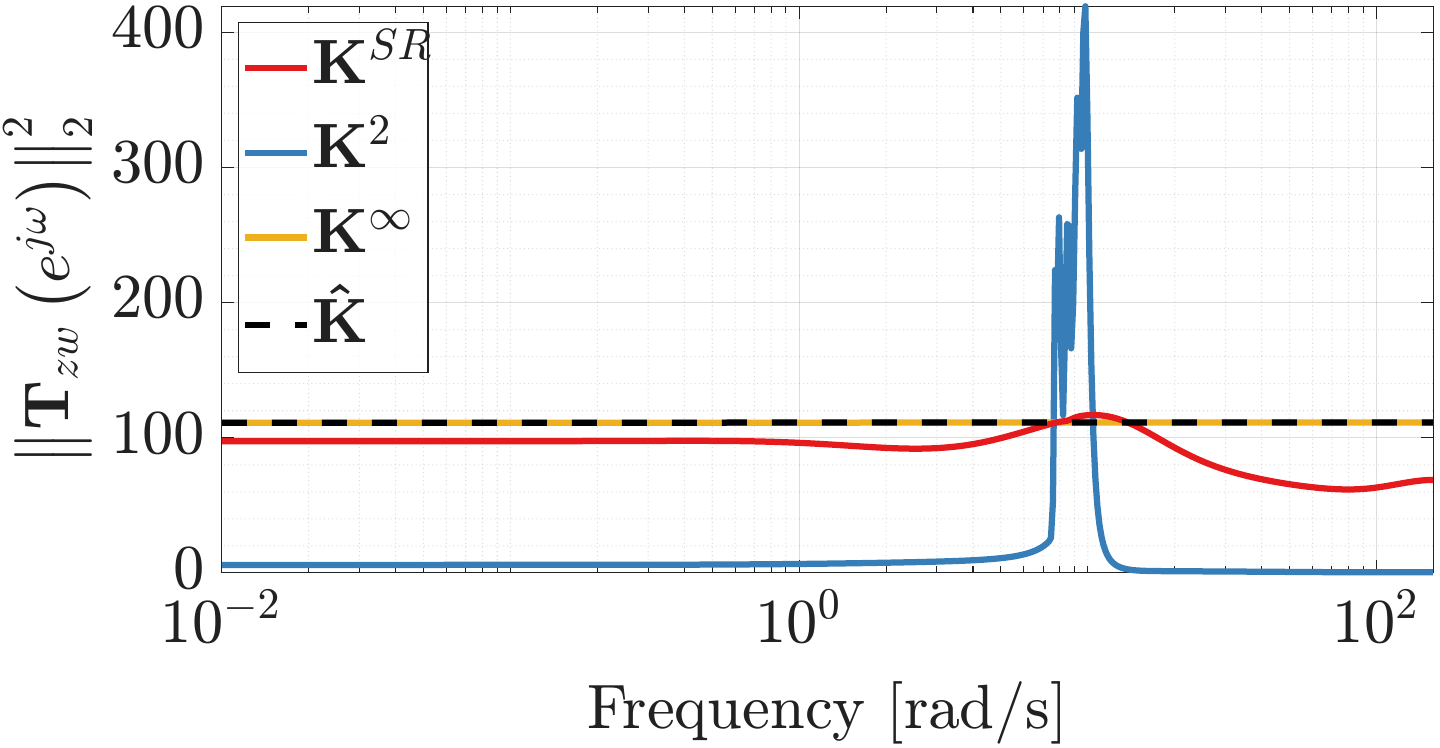}
    \caption{Plot of the squared 2-norm of $\tf{T}_{zw}(e^{j \omega})$ as a function of frequency.}
    \label{fig:plot_with_5_masses_on_all_channels}
\end{figure}

\begin{figure}[t]
    \centering
    \includegraphics[ width=0.9\linewidth]{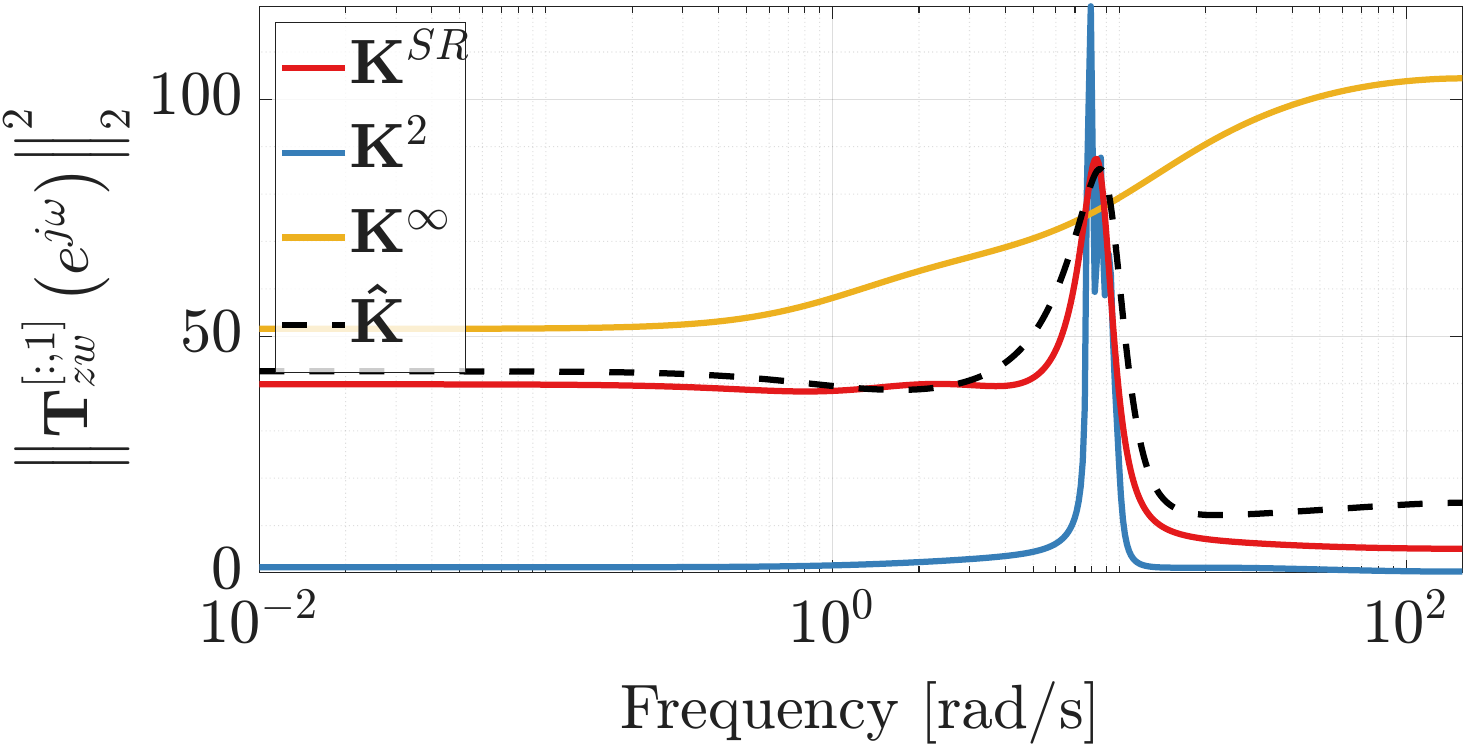}
    \caption{Plot of the squared 2-norm of $\tf{T}_{zw}^{[:,1]}(e^{j \omega})$ as a function of frequency.}
    \label{fig:plot_with_5_masses_on_1_channel}
\end{figure}
\begin{figure}[t]
    \centering
    \includegraphics[width=0.9\linewidth]{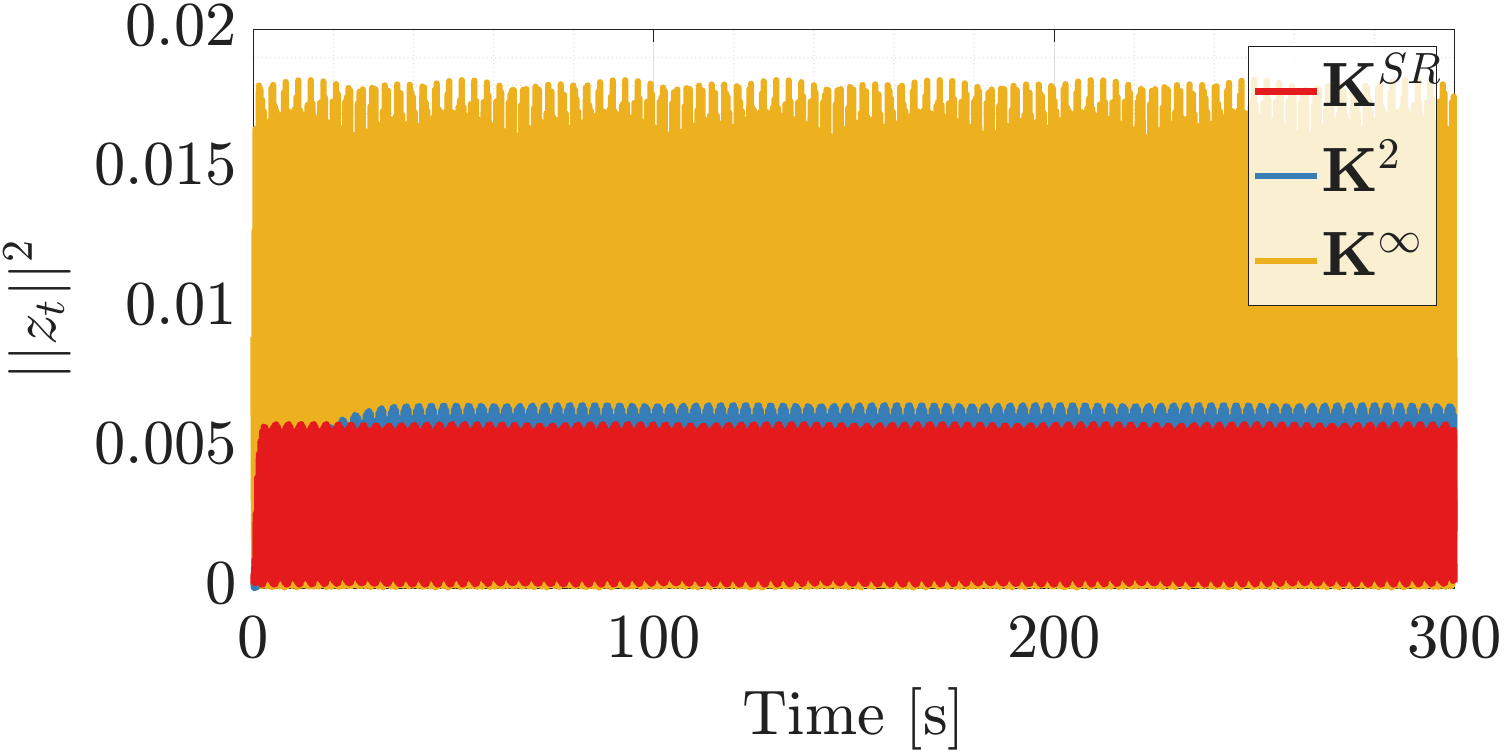}
    \caption{Plot of $\norm{z_t}^2$ as a function of time for the three controllers with a multi-sine disturbance with frequencies \SIlist{8;38}{\radian\per\sec} on the first bus.}%
    \label{fig:response_5_masses_sinusoidal_disturbance}
\end{figure}

\section{Conclusion} \label{sec:conclusion}
This paper developed a data-driven framework for distributed controller synthesis that minimises spatial regret, a graph-informed performance metric quantifying the cost of communication constraints. We extended the well-posedness result of the metric to arbitrary controller structures, removing restrictive assumptions on the parametrisation and oracle topology. The cost of this generalisation is that the synthesis problem must be solved iteratively, and the completeness of the set of left-factorised distributed stabilising controllers is not guaranteed.
By reformulating spatial regret minimisation in the frequency domain, we derived a tractable iterative convex optimisation algorithm that preserves closed-loop stability and operates directly on measured frequency-response data without requiring parametric system identification. Numerical experiments on a 5-bus power grid demonstrated that the spatial regret controller achieves superior disturbance rejection at targeted nodes compared to classical $\H2$ and $\Hinfty$ designs, with performance closely tracking that of the oracle controller.

Building upon the data-driven framework for robustness proposed in \citet{gupta_NonparametricIQCMultipliers_2026}, we aim to extend the present framework to address the challenge of noisy frequency-response measurements and nonlinearities in the system. Furthermore, the proposed synthesis algorithm is formulated as a semidefinite program, whose computational complexity scales poorly with network size. Consequently, exploring distributed implementations of the optimisation problem is another avenue of research to enable applications in large-scale systems.

\bibliography{bibliography}
                                                   
\appendix
\section{Equivalence of Fixed-Structure and Left-Factorised Distributed Controllers}
\label[app]{app:equivalence}

\begin{lem}
    Let $\mathcal{K}$ denote the set of admissible fixed-structure distributed controllers corresponding to the communication topology $\mathcal{S}$. Furthermore, let $\tilde{\mathcal{K}}$ denote the set of left factorised distributed controllers with the same communication topology $\mathcal{S}$.
    \begin{align*}
    \mathcal{K} 
        &= \ab\{
            \tf{K} \in \mathcal{R}_p^{n_o \times n_i}
            \,\middle\mid\,
            \tf{K} \in \operatorname{sparse}(\mathcal{S})
        \}
    \\
    \tilde{\mathcal{K}} 
        &= \ab\{
            \tf{Y}^{-1} \tf{X}
            \,\middle\mid\,
            \begin{matrix*}[l]
                \tf{Y} \in \mathcal{Y} \subseteq \mathcal{RH}_\infty \cap \operatorname{sparse}(\mathcal{D}) \\
                \tf{X} \in \mathcal{X} \subseteq \mathcal{RH}_\infty \cap \operatorname{sparse}(\mathcal{S})
            \end{matrix*}
        \}
    \end{align*}
    where $\operatorname{sparse}(\mathcal{D})$ describes a block diagonal structure. Then, the sets $\mathcal{K}$ and $\tilde{\mathcal{K}}$ are equivalent.
\end{lem}
\begin{pf}
    [${\tf{K} \in \mathcal{K} \implies \tf{K} \in \tilde{\mathcal{K}} }$] Let $\tf{K}_{i,j}$ denote the $(i,j)$-element of $\tf{K}$. We can express it as:
    \[
        \tf{K}_{i,j} = \frac{\tf{b}_{i,j}(z)}{\tf{a}_{i,j}(z)}
    \]
    where $\tf{b}_{i,j}(z)$ and $\tf{a}_{i,j}(z)$ are polynomials in the $z$-variable. We choose $\tf{Y} = \tf{y}(z) I$, where
    \[
        \tf{y}(z) = \frac{\prod_{i, j} \tf{a}_{i,j}(z)}{\ttf{\lambda}(z)}
    \]
    with $\ttf{\lambda}(z)$ being a stable polynomial of sufficiently high degree to ensure $\tf{Y} \in \mathcal{RH}_\infty$. For instance, one may choose $\ttf{\lambda}(z) = (z - \rho)^d$ for any $|\rho| < 1$, where $d$ is the degree of $\prod_{i,j} \tf{a}_{i,j}(z)$, ensuring that $\tf{y}(z)$ is proper and stable. Consequently, the $(i,j)$-element of $\tf{X}$ can be defined as ${\tf{X}_{i,j} = \tf{y} \tf{K}_{i,j}}$, which is also in $\mathcal{RH}_\infty$. Therefore, any controller ${\tf{K} \in \mathcal{K}}$ can be written as $\tf{Y}^{-1} \tf{X}$.
    
    [${\tf{K} \in \tilde{\mathcal{K}} \implies \tf{K} \in \mathcal{K} }$] This inclusion is established by the structural properties of $\mathcal{X}$ and $\mathcal{Y}$. \hfill{} \qed
\end{pf}

\end{document}